\documentclass[11pt]{article}

\usepackage{fullpage, amsmath, amsthm, amssymb, microtype, tikz, braket}
\usetikzlibrary{patterns}

\newtheorem{prop}{Proposition}
\newtheorem{thm}{Theorem}
\newtheorem{lem}{Lemma}

\newtheorem{cor}{Corollary}

\theoremstyle{definition}
\newtheorem{defn}{Definition}

\newcommand{\R}{\mathbb{R}}

\newcommand{\N}{\mathbb{N}}

\newcommand{\C}{\mathbb{C}}

\renewcommand{\epsilon}{\varepsilon}

\renewcommand{\hat}{\widehat}

\DeclareMathOperator*{\E}{E}

\DeclareMathOperator{\Tr}{Tr}

\title{Universal Bell Correlations Do Not Exist}
\author{Cole A. Graham\thanks{Department of Mathematics, Stanford University, \texttt{grahamca@stanford.edu}} \and William M. Hoza\thanks{Department of Computer Science, University of Texas at Austin, \texttt{whoza@utexas.edu}}}

\begin{document}
	\maketitle
	
	\begin{abstract}
		We prove that there is no finite-alphabet nonlocal box that generates exactly those correlations that can be generated using a maximally entangled pair of qubits. More generally, we prove that if some finite-alphabet nonlocal box is strong enough to simulate arbitrary local projective measurements of a maximally entangled pair of qubits, then that nonlocal box cannot itself be simulated using \emph{any} finite amount of entanglement. We also give a quantitative version of this theorem for approximate simulations, along with a corresponding upper bound.
	\end{abstract}

	\section{Introduction}
	A correlation box is a conceptual tool for reasoning about nonlocality:
	\begin{defn}
		A (discrete, bipartite) \emph{correlation box} is a map 
		\[
			\mathsf{Cor}: X \times Y \to \{\mathcal{D} : \mathcal{D} \text{ is a probability distribution over }A \times B\},
		\]
		where $X, Y, A, B$ are countable (finite or countably infinite) alphabets. We will abuse notation and write $\mathsf{Cor}: X \times Y \to A \times B$.
	\end{defn}
	Think of a correlation box $\mathsf{Cor}$ as a kind of ``channel'' through which two separated parties, Alice and Bob, can interact. Alice chooses $x \in X$ and Bob chooses $y \in Y$. A sample $(a, b)$ is drawn from $\mathsf{Cor}(x, y)$, and Alice is given $a$ and Bob is given $b$. The canonical example is the \emph{Popescu-Rohrlich box} $\mathsf{PR}: \{0, 1\} \times \{0, 1\} \to \{0, 1\} \times \{0, 1\}$, which is defined \cite{kt85, pr94} by
	\[
	\mathsf{PR}(x, y) \stackrel{\text{def}}{=} \begin{cases} (0, xy) & \text{with probability $1/2$} \\
	(1, 1 - xy) & \text{with probability $1/2$.} \end{cases}
	\]
	Observe that PR boxes cannot be used to communicate, since the marginal distributions of $a$ and $b$ are uniform regardless of $x$ and $y$. But $\mathsf{PR}$ is a \emph{nonlocal} box, i.e. given access to a PR box, Alice and Bob can perform tasks that would be impossible if they were isolated (even if they had shared randomness). The standard example is winning the CHSH game \cite{chsh69} with certainty.
	
	Qualitatively speaking, \emph{quantum entanglement is like a PR box}: it can be used to generate nonlocal correlations, but it cannot be used to communicate. Unfortunately, entanglement is not \emph{quantitatively} equivalent to a PR box; the Tsierelson bound \cite{cir80} implies that there is no quantum strategy for the CHSH game that wins with probability more than about $85\%$. In this work, we show that there is \emph{no} finite-alphabet correlation box that has exactly the same power as quantum entanglement.
	
	\subsection{Distributed sampling complexity classes}
	
	We can think of a correlation box as a \emph{distributed sampling problem} \cite{dlr05}: the problem of simulating the box. That is, Alice is given $x \in X$ and Bob is given $y \in Y$. Alice is supposed to output $a \in A$ and Bob is supposed to output $b \in B$ such that $(a, b) \sim \mathsf{Cor}(x, y)$.
	
	We define $\mathbf{SR}$ to be the class of all correlation boxes that can be simulated if Alice and Bob have unlimited shared randomness (but are otherwise isolated). We define $\mathbf{Q}$ to be the class of all correlation boxes that can be simulated if Alice and Bob have unlimited shared randomness and an arbitrary but finite amount of entanglement. Clearly, $\mathbf{SR} \subseteq \mathbf{Q}$. Bell's theorem \cite{bel64} can be interpreted as stating that $\mathbf{SR} \neq \mathbf{Q}$.
	
	For an upper bound on $\mathbf{Q}$, say that a correlation box $\mathsf{Cor}$ is \emph{non-signaling} if the marginal distribution of $a$ depends only on $x$ and the marginal distribution of $b$ depends only on $y$, where $(a, b) \sim \mathsf{Cor}(x, y)$. Let $\mathbf{NS}$ be the class of all non-signaling correlation boxes. In this notation, the \emph{no-communication theorem} states that $\mathbf{Q} \subseteq \mathbf{NS}$. The PR box shows that $\mathbf{Q} \neq \mathbf{NS}$. So to summarize, we have the proper inclusions $\mathbf{SR} \subsetneqq \mathbf{Q} \subsetneqq \mathbf{NS}$.
	
	\subsection{Our results}
	
	We define $\mathbf{BELL}$ to be the class of all correlation boxes that can be simulated if Alice and Bob have unlimited shared randomness, each holds one of a pair of maximally entangled qubits, and they are only allowed to make projective measurements. Understanding $\mathbf{BELL}$ is a good first step toward understanding $\mathbf{Q}$.
	
	Many previous results about simulating Bell correlations can be understood as \emph{reductions} between correlation boxes. A \emph{$k$-query reduction} from $\mathsf{Cor}_1$ to $\mathsf{Cor}_2$ is a protocol for simulating $\mathsf{Cor}_1$ in which Alice and Bob have unlimited shared randomness and $k$ copies of $\mathsf{Cor}_2$. (Taking a cue from quantum mechanics, we think of each correlation box as ``single use only''.) We will simply say that $\mathsf{Cor}_1$ \emph{reduces to} $\mathsf{Cor}_2$ if there is a $k$-query reduction from $\mathsf{Cor}_1$ to $\mathsf{Cor}_2$ for some $k$. (See Section~\ref{sec:reduction-model} for details.) We say that $\mathsf{Cor}: X \times Y \to A \times B$ is \emph{binary} if $X = Y = A = B = \{0, 1\}$. Our main result:
	
	\begin{thm} \label{thm:main}
		Suppose $\mathsf{Cor} \in \mathbf{Q}$ has countable input alphabets and finite output alphabets. Then there is some binary correlation box in $\mathbf{BELL}$ that does not reduce to $\mathsf{Cor}$.
	\end{thm}

	As usual, we say that $\mathsf{Cor}$ is \emph{$\mathbf{C}$-hard} if every correlation box in $\mathbf{C}$ reduces to $\mathsf{Cor}$. We say that $\mathsf{Cor}$ is \emph{$\mathbf{C}$-complete} if $\mathsf{Cor}$ is $\mathbf{C}$-hard and $\mathsf{Cor} \in \mathbf{C}$.

	\begin{cor}
		There does not exist a finite-alphabet $\mathbf{BELL}$-complete correlation box.
	\end{cor}

	\begin{cor}
		There does not exist a finite-alphabet $\mathbf{Q}$-complete correlation box.
	\end{cor}

	\begin{figure}
		\begin{center}
			\begin{tikzpicture}
				\draw (-4, 0) -- (4, 0);

				\begin{scope}
					\clip (-4.5, 0) rectangle (4.5, 8);
					\draw (0, 0) ellipse (1 and 1);
					\node at (0, 0.5) {$\mathbf{SR}$};
					
					\draw (0, 0) ellipse (2 and 3);
					\node at (1, 1.4) {$\mathbf{BELL}$};
					\draw (0, 0) ellipse (3 and 4.5);
					\node at (2, 2.2) {$\mathbf{Q}$};
					\draw (0, 0) ellipse (4 and 6);
					\node at (3, 3) {$\mathbf{NS}$};
					
					\draw (0, 10.2) ellipse (4 and 8);
					\node at (0, 7) {$\mathbf{BELL}$-hard};
					\fill (0, 5.2) circle (0.08) node[right] {$\mathsf{PR}$};
				\end{scope}
				
				\begin{scope}
					\clip (-4.5, 0) rectangle (4.5, 8);
					\clip (0, 10.2) ellipse (4 and 8);
					\fill[pattern=dots] (0, 0) ellipse (3 and 4.5);
				\end{scope}
			\end{tikzpicture}
		\end{center}
	
		\caption{Our result implies that the shaded region does not contain any finite-alphabet correlation boxes.}
	\end{figure}
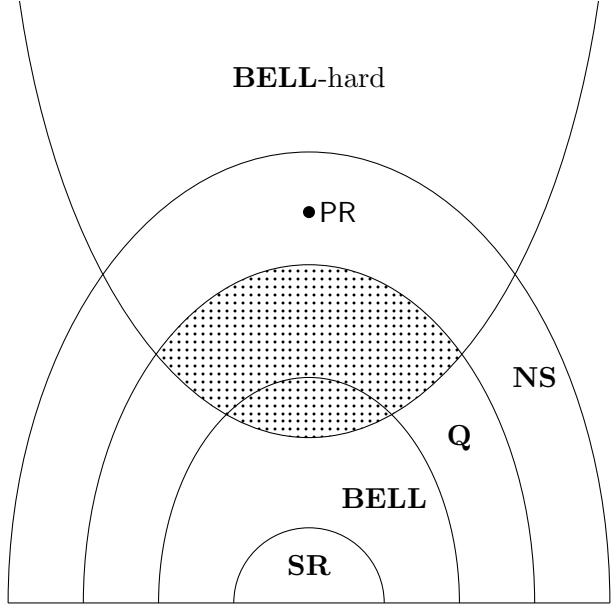

	Our result can be thought of as ``bad news'' for the project of understanding $\mathbf{BELL}$. We also give a quantitative version of our result for approximate simulations. An \emph{$\epsilon$-error reduction} is defined like an ordinary reduction except that we allow $\epsilon$ total variation error.
	
	\begin{thm} \label{thm:lower-bound}
		Suppose $\mathsf{Cor}_2: X \times Y \to A \times B$ is a finite-alphabet correlation box in $\mathbf{Q}$. Then there exists a binary correlation box $\mathsf{Cor}_1 \in \mathbf{BELL}$ such that for every $k$, if there is a $k$-query $\epsilon$-error reduction from $\mathsf{Cor}_1$ to $\mathsf{Cor}_2$, then
		\[
			k^4 \cdot (2|X|)^{4|A|^k} \cdot (2|Y|)^{4|B|^k} \geq \Omega(1/\epsilon).
		\]
	\end{thm}

	Conversely, for any $\epsilon > 0$, we give a simple construction of $\mathsf{Cor}: [T] \times [T] \to \{0, 1\} \times \{0, 1\}$ with $T \leq O(1/\epsilon^2)$ such that $\mathsf{Cor} \in \mathbf{BELL}$ and every correlation box in $\mathbf{BELL}$ reduces to $\mathsf{Cor}$ via a $1$-query $\epsilon$-error reduction. Notice that for $|A| = |B| = 2, k = 1$, Theorem~\ref{thm:lower-bound} implies that $|X| \cdot |Y|$ must be at least $1/\epsilon^{\Omega(1)}$. On the other hand, when $|A|, |B|, k$ are large, our lower bound might be very far from tight.

	\subsection{Related work}
	A long line of work \cite{mau92, bct99, ste00, cgm00, csi02, coa02, bt03, tb03} investigated the problem of simulating Bell correlations using classical communication, culminating in a theorem by Toner and Bacon \cite{tb03} that states that $\mathbf{BELL}$ can be simulated using shared randomness and a single classical bit of one-way communication. This result should be thought of as giving an \emph{upper bound} on the power of $\mathbf{BELL}$. Obviously it is a \emph{loose} upper bound, since $\mathbf{BELL} \subseteq \mathbf{NS}$.
	
	Cerf et al. \cite{cgmp05} improved on the Toner-Bacon theorem by showing that instead of a bit of communication, it suffices to have a single PR box. In our terminology, Cerf et al. showed that $\mathsf{PR}$ is $\mathbf{BELL}$-hard with respect to $1$-query reductions. Part of what makes this result so appealing is that $\mathsf{PR}$ has \emph{finite alphabets}, making it an extremely \emph{explicit} upper bound on $\mathbf{BELL}$. (Similarly with the Toner-Bacon theorem before it.)
	
	It is natural to hope to push even further and replace $\mathsf{PR}$ with some finite-alphabet correlation box in $\mathbf{Q}$. Our results dash this hope, even for the special case of simulating \emph{binary} correlation boxes in $\mathbf{BELL}$.
	
	In another direction, several works \cite{bp05, blm+05, dgh+07, bbl+06, fww09, bs09, vd13} have investigated the power of correlation boxes in their own right, apart from quantum entanglement. Two such works are particularly relevant to the present paper. First, Barrett and Pironio showed \cite{bp05} that every correlation box in $\mathbf{NS}$ with binary output alphabets reduces to $\mathsf{PR}$. Our result shows that there is no corresponding phenomenon for $\mathbf{BELL}$. Second, Dupuis et al. \cite{dgh+07} showed that no finite-alphabet correlation box is $\mathbf{NS}$-complete. Our result can be thought of as a ``scaled down'' version of this second result.
	
	\subsection{Proof overview} \label{sec:techniques}
	The \emph{biased CHSH game} is a variant of the well-studied CHSH game. In the biased game, Alice and Bob's input bits are not uniformly distributed. We will consider the case that their inputs are independent, Alice's is uniform, and Bob's has bias $p \in [1/2, 1]$. Alice and Bob know $p$, i.e. their strategy may depend on $p$. (See Section~\ref{sec:biased-chsh} for details.) We use a result by Lawson, Linden, and Popescu \cite{llp10} that states that the optimal quantum strategy for the biased CHSH game can be implemented in $\mathbf{BELL}$ and wins with probability $\frac{1}{2} + \frac{1}{2} \sqrt{p^2 + (1 - p)^2}$. Throughout this paper, we will let $\omega: \R \to \R$ denote this optimal success probability:
	\[
		\omega(p) \stackrel{\text{def}}{=} \frac{1}{2} + \frac{1}{2} \sqrt{p^2 + (1 - p)^2}.
	\]
	
	To prove Theorem~\ref{thm:main}, fix $\mathsf{Cor}: X \times Y \to A \times B$. Assume $\mathsf{Cor}$ is $\mathbf{BELL}$-hard; then for any $p$, there is some strategy for playing the biased CHSH game using finitely many copies of $\mathsf{Cor}$ that wins with probability $\omega(p)$. We can fix the shared randomness of the strategy without decreasing the probability of winning. Assume that $\mathsf{Cor} \in \mathbf{Q}$; then fixing the shared randomness must not have \emph{increased} the probability of winning. So the probability of winning is still exactly $\omega(p)$.
	
	But it is easy to show that for any deterministic strategy, the probability of winning is some affine function of $p$. If $X, Y$ are countable and $A, B$ are finite, there are only countably many deterministic strategies, and hence there are only countably many affine functions floating around. There must be some point $p$ where $\omega(p)$ disagrees with all of these affine functions, a contradiction.
	
	To prove our quantitative lower bound (Theorem~\ref{thm:lower-bound}), we extend the preceding argument by analyzing the distance between $\omega$ and any affine function at a randomly chosen point $p$.

	\subsection{Outline of this paper}
	In Section~\ref{sec:prelim}, we provide more detailed definitions of $\mathbf{SR}, \mathbf{BELL}, \mathbf{Q}$ and of our reduction model. In Section~\ref{sec:negative}, we prove our main, negative results. In Section~\ref{sec:chain}, we derive a simple consequence of our main result: there is an infinite chain of harder and harder finite-alphabet correlation boxes in $\mathbf{BELL}$. In Section~\ref{sec:positive}, we present our simple positive result. Finally, in Section~\ref{sec:open-problems}, we list some open problems.
	
	\section{Preliminaries} \label{sec:prelim}
	
	\subsection{Quantum and classical simulations}
	
	In this section, we give the technical definitions of $\mathbf{SR}, \mathbf{BELL}, \mathbf{Q}$. The reader who feels that these classes are intuitively clear may feel free to skip this section.
	
	Suppose $\mathcal{D}$ is a probability distribution over a class $\mathbf{C}$ of correlation boxes $X \times Y \to A \times B$. Then $\mathcal{D}$ induces a single correlation box $\mathsf{Cor}_\mathcal{D} : X \times Y \to A \times B$ defined by
	\begin{equation} \label{eqn:shared-randomness}
		\Pr[\mathsf{Cor}_{\mathcal{D}}(x, y) = (a, b)] = \E_{\mathsf{Cor} \sim \mathcal{D}}[\Pr[\mathsf{Cor}(x, y) = (a, b)]].
	\end{equation}
	(Intuitively, $\mathcal{D}$ models shared randomness. The distribution $\mathsf{Cor}_{\mathcal{D}}(x, y)$ is determined by sampling $\mathsf{Cor}$ from $\mathcal{D}$ and then using ``fresh randomness'' to sample $(a, b)$ from $\mathsf{Cor}(x, y)$.)
	
	Suppose $\mathbf{C}$ is a class of correlation boxes. We say that $\mathbf{C}$ is \emph{closed under convex combinations} if for every $X, Y, A, B$, for every distribution $\mathcal{D}$ over correlation boxes $X \times Y \to A \times B$ in $\mathbf{C}$, the box $\mathsf{Cor}_{\mathcal{D}}$ is also in $\mathbf{C}$.
	
	\begin{defn}
		We define $\mathbf{SR}$ to be the closure under convex combinations of the class of correlation boxes of the form $\mathsf{Cor}(x, y) = (f(x), g(y))$, where $f, g$ are (deterministic) functions.
	\end{defn}

	\begin{defn}
		We define $\mathbf{BELL}$ to be the closure under convex combinations of the class of correlation boxes $\mathsf{Cor}: X \times Y \to A \times B$ of the following form. For each $x \in X, y \in Y$, let $U_x, V_y$ be associated $2 \times 2$ unitary matrices. Define binary random variables $S_{x, y}, T_{x, y}$ by
		\[
			\Pr[(S_{x, y}, T_{x, y}) = (s, t)] = |\braket{st|U_x \otimes V_y|\phi}|^2,
		\]
		where $\ket{\phi} = \frac{1}{\sqrt{2}}(\ket{00} + \ket{11})$. Let $\mathsf{Cor}(x, y) = (f(S_{x, y}), g(T_{x, y}))$, where $f, g$ are (deterministic) functions.
	\end{defn}

	\begin{defn} \label{def:q}
		We define $\mathbf{Q}$ to be the closure under convex combinations of the class of correlation boxes $\mathsf{Cor}: X \times Y \to A \times B$ of the following form. Let $\rho$ be a bipartite mixed state on $\C^n \otimes \C^m$ for some finite $n, m$. For each $x \in X, y \in Y$, let $\{A_x^a\}_{a \in A}, \{B_y^b\}_{b \in B}$ be associated POVMs, where each $A_x^a$ acts on $\C^n$ and each $B_y^b$ acts on $\C^m$. The output distribution $\mathsf{Cor}(x, y)$ is defined by measuring $\rho$ using the POVMs associated with $x$ and $y$:
		\[
			\Pr[\mathsf{Cor}(x, y) = (a, b)] = \Tr((A_x^a \otimes B_y^b) \rho).
		\]
	\end{defn}

	Notice that Definition~\ref{def:q} allows for a $\mathbf{Q}$ protocol in which the shared quantum state $\rho$ is picked at random from some distribution over bipartite quantum states with finite Hilbert space dimensions.
	
	\subsection{Details of the reduction model} \label{sec:reduction-model}
	Let $\mathsf{Cor}_2: X_2 \times Y_2 \to A_2 \times B_2$ be a correlation box. In a \emph{deterministic $k$-query $\mathsf{Cor}_2$-protocol $\Pi: X_1 \times Y_1 \to A_1 \times B_1$}, Alice receives as input $x \in X_1$ and Bob receives $y \in Y_1$. The players make exactly $k$ queries $(x_1, y_1), \dots, (x_k, y_k)$ and get exactly $k$ responses $(a_1, b_1), \dots, (a_k, b_k)$, where each $x_i \in X_2, y_i \in Y_2, a_i \in A_2, b_i \in B_2$. The queries may be chosen adaptively, i.e. $x_i$ can be any deterministic function of $x, a_1, \dots, a_{i - 1}$ and $y_i$ can be any deterministic function of $y, b_1, \dots, b_{i - 1}$. The distribution of the responses is given by
	\[
		\Pr[(a_1, b_1), \dots, (a_k, b_k)] = \prod_{i = 1}^k \Pr[\mathsf{Cor}_2(x_i, y_i) = (a_i, b_i)],
	\]
	where $(x_i, y_i)$ is the $i$th query made by $\Pi$ when Alice and Bob see $(a_1, b_1), \dots, (a_{i - 1}, b_{i - 1})$ as the first $i - 1$ responses. At the end, Alice gives an output $a \in A_1$ and Bob gives an output $b \in B_1$. Here, $a$ is a deterministic function of $x, a_1, \dots, a_k$ and $b$ is a deterministic function of $y, b_1, \dots, b_k$. The distribution of $(a, b)$ as a function of $(x, y)$ defines a correlation box $\mathsf{Cor}_1: X_1 \times Y_1 \to A_1 \times B_1$; we say that $\Pi$ is a \emph{deterministic $k$-query reduction} from $\mathsf{Cor}_1$ to $\mathsf{Cor}_2$.

	A \emph{randomized $k$-query $\mathsf{Cor}_2$-protocol $\Pi: X_1 \times Y_1 \to A_1 \times B_1$} is a probability distribution $\mathcal{D}$ over deterministic $k$-query $\mathsf{Cor}_2$-protocols $\Pi': X_1 \times Y_1 \to A_1 \times B_1$. (This models ``shared randomness''.) This distribution induces a probability distribution $\mathcal{D}'$ over correlation boxes $X_1 \times Y_1 \to A_1 \times B_1$. Let $\mathsf{Cor}_1 = \mathsf{Cor}_{\mathcal{D}'}$, as defined in Equation~\ref{eqn:shared-randomness}. We say that $\Pi$ is a (randomized) \emph{$k$-query reduction} from $\mathsf{Cor}_1$ to $\mathsf{Cor}_2$.
	
	Suppose $\mathsf{Cor}_1, \mathsf{Cor}_1'$ are two correlation boxes on the same alphabets. We say that $\mathsf{Cor}_1$ is \emph{$\epsilon$-close} to $\mathsf{Cor}_1'$ if for every $x, y$, the distributions $\mathsf{Cor}_1(x, y)$, $\mathsf{Cor}_1'(x, y)$ are $\epsilon$-close in total variation distance. An \emph{$\epsilon$-error reduction} from $\mathsf{Cor}_1$ to $\mathsf{Cor}_2$ is a reduction from $\mathsf{Cor}_1'$ to $\mathsf{Cor}_2$ for some $\mathsf{Cor}_1'$ that is $\epsilon$-close to $\mathsf{Cor}_1$.
	
	\subsection{Closure}
	\begin{lem} \label{lem:closure}
		Suppose $\mathsf{Cor}_1$ reduces to $\mathsf{Cor}_2 \in \mathbf{Q}$. Then $\mathsf{Cor}_1 \in \mathbf{Q}$.
	\end{lem}

	\begin{proof}[Proof sketch]
		Say the reduction makes $k$ queries. The protocol witnessing $\mathsf{Cor}_2 \in \mathbf{Q}$ defines a probability distribution over bipartite quantum states. To simulate $\mathsf{Cor}_1$, Alice and Bob share $\rho_1 \otimes \dots \otimes \rho_k$, where the $\rho_i$s are drawn independently at random from that distribution. They run the reduction, using $\rho_i$ to simulate the $i$th query.
	\end{proof}
	
	We remark that $\mathbf{SR}$ and $\mathbf{NS}$ are also easily seen to be closed under reductions; $\mathbf{BELL}$ is closed under $1$-query reductions.
	
	\section{Negative results} \label{sec:negative}
	
	\subsection{The biased CHSH game} \label{sec:biased-chsh}
	
	For real numbers $p, q \in [0, 1]$, the \emph{biased CHSH game} $\mathrm{CHSH}[p, q]$ is a nonlocal game defined as follows \cite{llp10}: The referee picks $x, y \in \{0, 1\}$ independently at random, with $\Pr[x = 1] = p$, $\Pr[y = 1] = q$. Alice gets $x$ and Bob gets $y$. Alice outputs $a \in \{0, 1\}$ and Bob outputs $b \in \{0, 1\}$. The win condition is that $a + b = xy \pmod{2}$. The standard CHSH game \cite{chsh69} is the case $p = q = \frac{1}{2}$.
	
	We can think of a correlation box $\mathsf{Cor}: \{0, 1\} \times \{0, 1\} \to \{0, 1\} \times \{0, 1\}$ as a \emph{strategy} for the biased CHSH game. The probability that $\mathsf{Cor}$ wins $\mathrm{CHSH}[p, q]$ is just the probability that $a + b = xy \pmod{2}$, where $(a, b) = \mathsf{Cor}(x, y)$ and the probability is over both the internal randomness of $\mathsf{Cor}$ and the inputs $x, y$. (The inputs $(x, y)$ are independent of the internal randomness of $\mathsf{Cor}$.)
	
	Lawson et al. showed that like in the ordinary CHSH game, quantum entanglement gives an advantage in the biased CHSH game, at least in certain parameter regimes:
	\begin{lem}[\cite{llp10}] \label{lem:biased-chsh-positive}
		If $\frac{1}{2} \leq q \leq \frac{1}{2p} \leq 1$, then there exists a binary-alphabet correlation box $\mathsf{S}_{p, q} \in \mathbf{BELL}$ that wins $\mathrm{CHSH}[p, q]$ with probability $\frac{1}{2} + \frac{1}{2} \sqrt{2} \sqrt{q^2 + (1 - q)^2} \sqrt{p^2 + (1 - p)^2}$.
	\end{lem}
	
	Conversely, Lawson et al. also showed that Lemma~\ref{lem:biased-chsh-positive} is optimal:
	\begin{lem}[\cite{llp10}] \label{lem:biased-chsh-negative}
		Suppose $\frac{1}{2} \leq q \leq \frac{1}{2p} \leq 1$ and $\mathsf{Cor} \in \mathbf{Q}$. Then $\mathsf{Cor}$ wins $\mathrm{CHSH}[p, q]$ with probability at most $\frac{1}{2} + \frac{1}{2} \sqrt{2} \sqrt{q^2 + (1 - q)^2} \sqrt{p^2 + (1 - p)^2}$.
	\end{lem}

	\subsection{Reductions imply affine approximations}
	
	\begin{lem} \label{lem:box-lines}
		Suppose $\mathsf{Cor}: X \times Y \to A \times B$ is a correlation box in $\mathbf{Q}$, and fix $k \in \N$. For each $p \in [1/2, 1]$, let $\mathsf{S}_{p, 1/2}$ be the box of Lemma~$\ref{lem:biased-chsh-positive}$. There is some set $L_{\mathsf{Cor}, k}$ of affine functions $\R \to \R$ such that:
		\begin{enumerate}
			\item For every $p \in [1/2, 1]$ and every $\epsilon > 0$, if there exists a $k$-query $\epsilon$-error reduction from $\mathsf{S}_{p, 1/2}$ to $\mathsf{Cor}$, then there exists $\ell \in L_{\mathsf{Cor}, k}$ such that $|\ell(p) - \omega(p)| \leq \epsilon$.
			\item If $X, Y, A, B$ are all finite, then
			\[
			|L_{\mathsf{Cor}, k}| \leq (2|X|)^{2|A|^k} \cdot (2|Y|)^{2|B|^k}.
			\]
			If $X, Y$ are countable and $A, B$ are finite, then $L_{\mathsf{Cor}, k}$ is countable.
		\end{enumerate}
	\end{lem}
	
	\begin{proof}
		For a deterministic $\mathsf{Cor}$-protocol $\Pi$, let $\ell_{\Pi}(p)$ be the probability that $\Pi$ wins $\mathrm{CHSH}[p, 1/2]$. Then $\ell_{\Pi}$ is an affine function, since it is just
		\[
		\frac{1 - p}{2} P_{00} + \frac{p}{2} P_{10} + \frac{1 - p}{2} P_{01} + \frac{p}{2} P_{11},
		\]
		where $P_{xy}$ is the probability that $a + b = xy \pmod{2}$ where $(a, b) = \Pi(x, y)$. Let $L_{\mathsf{Cor}, k}$ be the set of all $\ell_{\Pi}$.
		
		To prove the first item, let $\Pi$ be a $k$-query $\epsilon$-reduction from $S_{p, 1/2}$ to $\mathsf{Cor}$. Recall that $\Pi$ is a distribution over deterministic $\mathsf{Cor}$-protocols $\Pi'$. Let $g(\Pi')$ be the probability that $\Pi'$ wins $\mathrm{CHSH}[p, 1/2]$. By the correctness of the reduction, we know that
		\[
		\left|\E_{\Pi' \sim \Pi}[g(\Pi')] - \omega(p)\right| \leq \epsilon.
		\]
		The best case is at least as good as the average case, so there exists a deterministic $\mathsf{Cor}$-protocol $\Pi'_*$ such that $g(\Pi'_*) \geq \omega(p) - \epsilon$. Since $\mathsf{Cor} \in \mathbf{Q}$, by Lemma~\ref{lem:closure}, $\Pi'_*$ implements a correlation box in $\mathbf{Q}$. Therefore, by Lemma~\ref{lem:biased-chsh-negative}, $g(\Pi'_*) \leq \omega(p)$. Therefore, $|g(\Pi'_*) - \omega(p)| \leq \epsilon$. By the construction of $L_{\mathsf{Cor}, k}$, there is some $\ell \in L_{\mathsf{Cor}, k}$ such that $g(\Pi'_*) = \ell(p)$, and hence $|\ell(p) - \omega(p)| \leq \epsilon$.
		
		To prove the second item, we bound the cardinality of $L_{\mathsf{Cor}, k}$ simply by bounding the number of deterministic $k$-query $\mathsf{Cor}$-protocols. Such a protocol can be specified by:
		\begin{itemize}
			\item Functions $q_i: \{0, 1\} \times A^{i - 1} \to X$ for each $1 \leq i \leq k$, telling the $i$th query that Alice makes as a function of her input and the query responses she has seen so far.
			\item Corresponding functions $r_i: \{0, 1\} \times B^{i - 1} \to Y$ for Bob.
			\item A function $s: \{0, 1\} \times A^k \to \{0, 1\}$, telling the output Alice gives as a function of her input and all query responses.
			\item A corresponding function $t: \{0, 1\} \times B^k \to \{0, 1\}$ for Bob.
		\end{itemize}
		If $X, Y$ are countable and $A, B$ are finite, then there are only countably many possibilities for each of these functions, so there are countably many such protocols. Suppose now that $X, Y, A, B$ are all finite and $|A|, |B| \geq 2$. The number of possible functions $q_i$ is $|X|^{2|A|^{i - 1}}$, and similarly for $r_i$. The number of possible functions $s$ is $2^{2|A|^k}$, and similarly for $t$. Therefore, the number of affine functions is bounded by
		\begin{align*}
		\left(\prod_{i = 1}^k |X|^{2|A|^{i - 1}}\right) \left(\prod_{i = 1}^k |Y|^{2|B|^{i - 1}}\right) \cdot 2^{2|A|^k} \cdot 2^{2|B|^k} &= |X|^{2 \sum_i |A|^{i - 1}} \cdot |Y|^{2 \sum_i |B|^{i - 1}} \cdot 2^{2|A|^k} \cdot 2^{2|B|^k} \\
		&\leq |X|^{2|A|^k} \cdot |Y|^{2|B|^k} \cdot 2^{2|A|^k} \cdot 2^{2|B|^k} \\
		&= (2|X|)^{2|A|^k} \cdot (2|Y|)^{2|B|^k}.
		\end{align*}
		Finally, if $A$ is a singleton set, the step above where we bounded the geometric series $\sum_i |A|^{i - 1}$ by $|A|^k$ was not valid, but in this case the functions $q_i$ do not need to be specified anyway, so the final bound still holds. Similarly if $B$ is a singleton set.
	\end{proof}

	\subsection{Lower bounds on the error of affine approximations}
	
	For our qualitative negative result (Theorem~\ref{thm:main}), the following trivial fact is sufficient.
	
	\begin{lem} \label{lem:countable-lines-approximation}
		Suppose $L$ is a countable set of affine functions $\R \to \R$. Then there is some $p \in [1/2, 1]$ such that for every $\ell \in L$, $\ell(p) \neq \omega(p)$.
	\end{lem}

	\begin{proof}
		Suppose that some value of $p$ satisfies $\ell(p) = \omega(p)$, where $\ell \in L$. Rearranging, we find that
		\begin{equation} \label{eqn:disc}
			p^2 + (1 - p)^2 = r(p)^2,
		\end{equation}
		where $r(p)$ is another affine function. The quadratic expression on the left hand side of Equation~\ref{eqn:disc} has a nonzero discriminant of $-4$. Therefore, Equation~\ref{eqn:disc} must not be an identity, and hence it has at most two solutions $p$. So each $\ell \in L$ intersects $\omega$ at most twice, and hence $L$ intersects $\omega$ in at most countably many places.
	\end{proof}

	For our quantitative negative result (Theorem~\ref{thm:lower-bound}), we need to lower bound the error of any approximation of $\omega$ by affine functions.
	
	\begin{lem} \label{lem:line-approximation-error}
		Pick $p \in [1/2, 1]$ uniformly at random. Then for any affine function $\ell: \R \to \R$ and any $\epsilon > 0$,
		\[
		\Pr\left[|\ell(p) - \omega(p)| \leq \epsilon\right] \leq O(\sqrt{\epsilon}).
		\]
	\end{lem}
	
	\begin{proof}
		Let $I = [1/2, 1]$. We first compute
		\begin{equation} \label{eqn:second-deriv}
			\omega''(x) = \frac{1}{2}[x^2 + (1 - x)^2]^{-3/2} = \frac{1}{2}\left[2\left(x - \frac{1}{2}\right)^2 + \frac{1}{2}\right]^{-3/2} \geq \frac{1}{2} \quad \text{on $I$.}
		\end{equation}
		Hence $\omega$ is uniformly convex on $I$.
		
		Without loss of generality we can assume that the graph of $\ell$ intersects the graph of $\omega$ twice (with a point of tangency counted as a double intersection). After all, if $\ell < \omega$ on $I$, translate $\ell$ up until the first moment of equality with $\omega$, thus decreasing the pointwise error between $\ell$ and $\omega$ at every $x \in I$. If $\ell$ is then tangent to $\omega$, we are done. Otherwise, $\ell$ intersects $\omega$ at an endpoint, so rotate $\ell$ up about this point until it is tangent to $\omega$ (no other intersections occur because $\omega$ is uniformly convex). Pointwise errors do not increase under this rotation, so the entire transformation only increases the probability in the lemma statement. Similar considerations hold if initially $\ell > \omega$ or $\ell$ intersects $\omega$ at one point. Thus we may assume that $\ell$ linearly interpolates $\omega$.
		
		Suppose $\ell$ interpolates $\omega$ at the (potentially coincident) points $x_1, x_2 \in I$. We now claim that for all $x \in I$, there exists $\xi_x \in I$ such that
		\begin{equation} \label{eqn:mvt}
			\omega(x) - \ell(x) = \frac{\omega''(\xi_x)}{2} (x - x_1)(x - x_2).
		\end{equation}
		This follows from a standard argument in interpolation theory; we include the details here for completeness. If $x = x_1$ or $x = x_2$, Equation~\ref{eqn:mvt} is trivial, since both sides are zero. Otherwise, let $\phi_x(t) = \omega(t) - \ell(t) - (\omega(x) - \ell(x)) \cdot \frac{(t - x_1)(t - x_2)}{(x - x_1)(x - x_2)}$. Then $\phi_x$ is zero at $x$, $x_1$, and $x_2$. By Rolle's theorem, this implies that $\phi_x'$ has at least two zeroes in $I$ (actually Rolle's theorem only gives one zero if $x_1 = x_2$, but in this case $x_1 = x_2$ is another zero of $\phi_x'$, so either way $\phi_x'$ has two distinct zeroes in $I$). By another application of Rolle's theorem, there is some $\xi_x \in I$ such that $\phi_x''(\xi_x) = 0$. Equation~\ref{eqn:mvt} follows.
		
		By Equation~\ref{eqn:second-deriv},
		\[
			|\omega(x) - \ell(x)| \geq \frac{1}{4} |x - x_1| |x - x_2|.
		\]
		In particular, when $\min\{|x - x_1|, |x - x_2|\} > 2\sqrt{\epsilon}$, $|\omega(x) - \ell(x)| > \epsilon$. The probability that $p$ is within $2\sqrt{\epsilon}$ of either $x_1$ or $x_2$ is $O(\sqrt{\epsilon})$ by the union bound.
	\end{proof}
	
	\subsection{Proofs of main results}
	
	\begin{proof}[Proof of Theorem~$\ref{thm:main}$]
		Fix $\mathsf{Cor}: X \times Y \to A \times B$, where $X, Y$ are countable, $A, B$ are finite, and $\mathsf{Cor} \in \mathbf{Q}$. We will show that there is some choice of $p$ so that there is no reduction from $\mathsf{S}_{p, 1/2}$ to $\mathsf{Cor}$; since $\mathsf{S}_{p, 1/2}$ is a binary correlation box in $\mathbf{BELL}$, this will complete the proof.
		
		For each $k \in \N$, let $L_{\mathsf{Cor}, k}$ be the set of affine functions given by Lemma~\ref{lem:box-lines}. The alphabet bounds for $\mathsf{Cor}$ imply that $L_{\mathsf{Cor}, k}$ is countable. Let $L = \bigcup_{k \in \N} L_{\mathsf{Cor}, k}$, so that $L$ is still countable. By Lemma~\ref{lem:countable-lines-approximation}, choose $p \in [1/2, 1]$ so that for every $\ell \in L$, $\ell(p) \neq \omega(p)$. Then $\mathsf{S}_{p, 1/2}$ does not reduce to $\mathsf{Cor}$, because if there were a $k$-query ($0$-error) reduction for some $k$, Lemma~\ref{lem:box-lines} would imply that there was some $\ell \in L$ with $\ell(p) = \omega(p)$.
	\end{proof}

	\begin{proof}[Proof of Theorem~$\ref{thm:lower-bound}$]
		Fix $\mathsf{Cor}_2: X \times Y \to A \times B$, where $X, Y, A, B$ are finite and $\mathsf{Cor}_2 \in \mathbf{Q}$. Let $L_{\mathsf{Cor}_2, k}$ be the set of affine functions given by Lemma~\ref{lem:box-lines}. Pick $p \in [1/2, 1]$ uniformly at random. By Lemma~\ref{lem:line-approximation-error} and the union bound, for any $\epsilon_k > 0$, the probability that some $\ell \in L_{\mathsf{Cor}_2, k}$ satisfies $|\ell(p) - \omega(p)| \leq \epsilon_k$ is at most $O(\sqrt{\epsilon_k} \cdot |L_{\mathsf{Cor}_2, k}|)$. Therefore, by the union bound over $k$,
		\begin{equation} \label{eqn:prob-bound}
			\Pr[\exists k,  \exists \ell \in L_{\mathsf{Cor}_2, k}, |\ell(p) - \omega(p)| \leq \epsilon_k] \leq O\left(\sum_{k = 1}^{\infty} \sqrt{\epsilon_k} \cdot |L_{\mathsf{Cor}_2, k}|\right).
		\end{equation}
		Choose $\epsilon_k$ so that $\sqrt{\epsilon_k} \cdot |L_{\mathsf{Cor}_2, k}| = c/k^2$, where $c$ is a sufficiently small constant so that the bound in Equation~\ref{eqn:prob-bound} is strictly less than $1$. (Such a $c$ exists because $\sum_k 1/k^2$ is a convergent series.) This can be achieved while maintaining
		\[
			\epsilon_k \geq \Omega\left(\frac{1}{k^4 |L_{\mathsf{Cor}_2, k}|^2}\right),
		\]
		which implies by Lemma~\ref{lem:box-lines} that
		\[
			k^4 \cdot (2|X|)^{4|A|^k} \cdot (2|Y|)^{4|B|^k} \geq \Omega(1/\epsilon_k).
		\]
		
		By our choice of $\epsilon_k$, there exists some $p$ so that for every $k$, for every $\ell \in L_{\mathsf{Cor}_2, k}$, $|\ell(p) - \omega(p)| > \epsilon_k$. Choose $\mathsf{Cor}_1 = \mathsf{S}_{p, 1/2}$. By Lemma~\ref{lem:box-lines}, if there is a $k$-query $\epsilon$-error reduction from $\mathsf{Cor}_1$ to $\mathsf{Cor}_2$, then $\epsilon > \epsilon_k$.
	\end{proof}

	\section{There's always a harder box} \label{sec:chain}
	The proofs of our main results are done. In this section, we make a simple observation that follows easily from our negative results. Define a preorder on correlation boxes by saying that $\mathsf{Cor} \leq \mathsf{Cor}'$ if there is a reduction from $\mathsf{Cor}$ to $\mathsf{Cor}'$. Write $\mathsf{Cor} < \mathsf{Cor}'$ if $\mathsf{Cor} \leq \mathsf{Cor}'$ and $\mathsf{Cor}' \not \leq \mathsf{Cor}$.
	\begin{thm}
		For any finite-alphabet correlation box $\mathsf{Cor} \in \mathbf{BELL}$, there is another finite-alphabet correlation box $\mathsf{Cor}' \in \mathbf{BELL}$ such that $\mathsf{Cor} < \mathsf{Cor}'$.
	\end{thm}

	\begin{proof}
		By Theorem~\ref{thm:main}, there is a binary correlation box $\mathsf{Cor}_0 \in \mathbf{BELL}$ such that $\mathsf{Cor}_0 \not \leq \mathsf{Cor}$. Write $\mathsf{Cor}: X \times Y \to A \times B$. By relabeling if necessary, we can assume that $0, 1 \not \in X, Y$. Define
		\[
			\mathsf{Cor}': (X \cup \{0, 1\}) \times (Y \cup \{0, 1\}) \to (A \cup \{0, 1\}) \times (B \cup \{0, 1\})
		\]
		by the following $\mathbf{BELL}$ algorithm:
		\begin{itemize}
			\item If $x \in X$, then Alice does what she would have done in the protocol witnessing $\mathsf{Cor} \in \mathbf{BELL}$. Otherwise, if $x \in \{0, 1\}$, she does what she would have done in the protocol witnessing $\mathsf{Cor}_0 \in \mathbf{BELL}$.
			\item Bob acts similarly.
		\end{itemize}
		By construction:
		\begin{itemize}
			\item If $x \in X, y \in Y$, then $\mathsf{Cor}'(x, y) \sim \mathsf{Cor}(x, y)$. This immediately implies that $\mathsf{Cor} \leq \mathsf{Cor}'$.
			\item If $x, y \in \{0, 1\}$, then $\mathsf{Cor}'(x, y) \sim \mathsf{Cor}_0(x, y)$. This immediately implies that $\mathsf{Cor}_0 \leq \mathsf{Cor}'$, and hence by transitivity $\mathsf{Cor}' \not \leq \mathsf{Cor}$.
		\end{itemize}
		(Notice that if $x \in X, y \in \{0, 1\}$, the distribution $\mathsf{Cor}'(x, y)$ has no clear interpretation, but that doesn't matter for us. Similarly with the case $x \in \{0, 1\}, y \in Y$.)
	\end{proof}

	\section{Positive results} \label{sec:positive}

	We now show how to construct a finite-alphabet correlation box that is approximately complete for $\mathbf{BELL}$. The construction is simple, and just consists of an appropriate discretization of the Bloch sphere \cite{blo46}.

	\begin{thm}
		For every $\epsilon > 0$, there exists $\mathsf{Cor}_2: [T] \times [T] \to \{0, 1\} \times \{0, 1\}$ with $T \leq O(1/\epsilon^2)$ such that $\mathsf{Cor}_2 \in \mathbf{BELL}$, and for every $\mathsf{Cor}_1 \in \mathbf{BELL}$, there is a $1$-query $\epsilon$-error reduction from $\mathsf{Cor}_1$ to $\mathsf{Cor}_2$.
	\end{thm}

	\begin{proof}
		Let $c_1, c_2, \dots, c_T \in \R^3$ be points on the unit sphere such that every point on the unit sphere is within $\epsilon$ of some $c_i$ in $\ell_2$ distance. Such a collection of points exists with $T \leq O(1/\epsilon^2)$. We define $\mathsf{Cor}_2$ by the following algorithm, simultaneously showing that $\mathsf{Cor}_2 \in \mathbf{BELL}$: Alice and Bob share a pair of qubits in the state $\ket{\phi} \stackrel{\text{def}}{=} \frac{\ket{01} - \ket{10}}{\sqrt{2}}$. (This can be obtained by applying local operations to $\frac{\ket{00} + \ket{11}}{\sqrt{2}}$.) On inputs $i, j$:
		\begin{enumerate}
			\item Alice finds a unitary matrix $U$ such that $U^{-1}\ket{1}$ is represented by the point $c_i$ on the Bloch sphere. She applies $U$ to her qubit, measures in the computational basis, and outputs the observed bit.
			\item Bob finds a unitary matrix $V$ such that $V^{-1}\ket{1}$ is represented by the point $c_j$ on the Bloch sphere. He applies $V$ to his qubit, measures in the computational basis, and outputs the observed bit.
		\end{enumerate}
		We now give the reduction. From the definition of $\mathbf{BELL}$, it suffices to show how to approximately simulate applying some unitary matrix $U \otimes V$ to $\ket{\phi}$ and then measuring in the computational basis, where Alice chooses $U$ and Bob chooses $V$. To do this, Alice finds $c_i$ that is closest to the Bloch sphere representation of $U^{-1} \ket{1}$ in $\ell_2$ distance, and Bob finds $c_j$ that is closest to the Bloch sphere representation of $V^{-1} \ket{1}$ in $\ell_2$ distance. They query $\mathsf{Cor}_2(i, j)$.
		
		We now prove correctness of this reduction. A curiosity of the state $\ket{\phi}$ is that for \emph{any} unitary $V$, there is a scalar $\lambda \in \C$ such that $(V \otimes V) \ket{\phi} = \lambda \ket{\phi}$. Proof:
		\begin{align*}
			\braket{ij|V \otimes V|\phi} = \frac{\braket{i|V|0} \braket{j|V|1} - \braket{i|V|1} \braket{j|V|0}}{\sqrt{2}},
		\end{align*}
		and hence $\braket{00|V \otimes V|\phi} = \braket{11|V \otimes V|\phi} = 0$ and $\braket{01|V \otimes V|\phi} = -\braket{10|V \otimes V|\phi}$.
		
		Therefore, we can write
		\begin{align*}
			(U \otimes V)\ket{\phi} &= (UV^{-1} \otimes I) (V \otimes V) \ket{\phi} \\
			&=\lambda (UV^{-1} \otimes I) \ket{\phi} \\
			&= \frac{\lambda}{\sqrt{2}}((UV^{-1} \ket{0}) \ket{1} - (UV^{-1} \ket{1}) \ket{0}).
		\end{align*}
		It follows that when $(U \otimes V) \ket{\phi}$ is measured, giving two bits $a, b$,
		\begin{align*}
			\Pr[a = b] &= \frac{1}{2} |\braket{1| UV^{-1} | 0}|^2 + \frac{1}{2} |\braket{0|UV^{-1}|1}|^2 \\
			&= 1 - |\braket{1 | UV^{-1} | 1}|^2.
		\end{align*}
		Let $x$ be the Bloch sphere representation of $U^{-1} \ket{1}$, and let $y$ be the Bloch sphere representation of $V^{-1} \ket{1}$. Then $|\braket{1 | UV^{-1} | 1}|^2 = \frac{1}{2} + \frac{1}{2} x \cdot y$, where $\cdot$ is the dot product. So $\Pr[a = b] = \frac{1}{2} - \frac{1}{2} x \cdot y$. Let $(\hat{a}, \hat{b}) = \mathsf{Cor}(i, j)$. Then
		\begin{align*}
			\left|\Pr[a = b] - \Pr[\hat{a} = \hat{b}]\right| &= \frac{1}{2} |x \cdot y - c_i \cdot c_j| \\
			&\leq \frac{1}{2} |x \cdot y - x \cdot c_j| + \frac{1}{2}|x \cdot c_j - c_i \cdot c_j| \\
			&\leq \epsilon/2 + \epsilon/2.
		\end{align*}
		Since $a, b, \hat{a}, \hat{b}$ all have uniform marginal distributions, it follows that $(a, b)$ and $(\hat{a}, \hat{b})$ are $\epsilon$-close in total variation distance.
	\end{proof}

	\begin{prop}
		There exists $\mathsf{Cor}_2: \N \times \N \to \{0, 1\} \times \{0, 1\}$ such that $\mathsf{Cor}_2 \in \mathbf{BELL}$, and for every $\mathsf{Cor}_1 \in \mathbf{BELL}$ and every $\epsilon > 0$, there is a $1$-query $\epsilon$-error reduction from $\mathsf{Cor}_1$ to $\mathsf{Cor}_2$.
	\end{prop}

	\begin{proof}[Proof sketch]
		Use a countable dense subset of the Bloch sphere.
	\end{proof}

	\section{Open problems} \label{sec:open-problems}

	\begin{itemize}
		\item We proved that there is no $\mathbf{BELL}$-complete correlation box with countable input alphabets and \emph{finite} output alphabets. Does there exist a $\mathbf{BELL}$-complete correlation box with countable alphabets? (Our proof breaks down because there are uncountably many deterministic reductions to a correlation box with countably infinite output alphabets.)
		\item Does there exist a \emph{minimal} $\mathbf{BELL}$-hard finite-alphabet correlation box $\mathsf{Cor}$? (By minimal, we mean that if $\mathsf{Cor}'$ is another $\mathbf{BELL}$-hard finite-alphabet correlation box, then $\mathsf{Cor}$ reduces to $\mathsf{Cor}'$.) 
		\item What is the right relationship between $|X|, |Y|, |A|, |B|, k, \epsilon$ in Theorem~\ref{thm:lower-bound}?
	\end{itemize}

	\section{Acknowledgments}
	We thank Scott Aaronson and Ronald de Wolf for helpful comments and encouragement. This material is based upon work supported by the National Science Foundation Graduate Research Fellowship under Grant No. DGE-1610403. Cole Graham gratefully acknowledges the support of the Fannie and John Hertz Foundation.

	\bibliographystyle{alpha}
	\bibliography{bell}
\end{document}